\documentclass[11pt]{article}
\usepackage{amsmath,amsthm,amsfonts,graphicx,caption,float,color,soul,authblk,fullpage}
\usepackage[colorlinks=true,allcolors=blue]{hyperref}

\usepackage{multirow}
\usepackage[ruled,vlined,linesnumbered]{algorithm2e}
\usepackage{xspace}

\newtheorem{thm}{Theorem}
\newtheorem{lem}[thm]{Lemma}

\newtheorem{dfn}[thm]{Definition}

\def\RR{{\mathbb R}}
\def\Ss{{\mathbb S}}
\def\EE{{\mathbb E}}
\def\Pr{{\mathrm Pr}}
\def\dd{{\mathrm d}}
\def\ee{{\mathrm e}}
\def\MM{{\mathcal M}}
\def\sign{{\mathrm sign}}

\begin{document}
\title{Practical linear-space Approximate Near Neighbors in high dimension}
\author[1,2]{Georgia Avarikioti}
\author[2]{Ioannis Z.~Emiris}
\author[2]{Ioannis Psarros}
\author[2]{Georgios Samaras}
\affil[1]{School of Electrical and Computer Engineering, National Technical University of Athens,\\ Athens, Greece\\
  \texttt{zetavar@hotmail.com}}
\affil[2]{Department of Informatics \& Telecommunications, National Kapodistrian University of Athens,\\ Athens, Greece\\
  \texttt{\{emiris,ipsarros,gsamaras\}@di.uoa.gr}}

\maketitle

\begin{abstract}
The $c$-approximate Near Neighbor problem in high dimensional spaces has been mainly addressed by Locality Sensitive Hashing (LSH), which offers polynomial dependence on the dimension, query time sublinear in the size of the dataset, and subquadratic space requirement. 
For practical applications, linear space is typically imperative.
Most previous work in the linear space regime focuses on the case that $c$ exceeds $1$ by a constant term. In a recently accepted paper, optimal bounds have been achieved for any $c>1$~\cite{ALRW17}.

Towards practicality, we present a new and simple data structure using linear space and sublinear query time for any $c>1$ including $c\to 1^+$. Given an LSH family of functions for some metric space, we randomly project points to the Hamming cube of dimension $\log n$, where $n$ is the number of input points. The projected space contains strings which serve as keys for buckets containing the input points. The query algorithm simply projects the query point, then examines points which are assigned to the same or nearby vertices on the Hamming cube. 
We analyze in detail the query time for some standard LSH families. 

To illustrate our claim of practicality, we offer an open-source implementation in {\tt C++}, and report on several experiments in dimension up to 1000 and $n$ up to $10^6$. 
Our algorithm is one to two orders of magnitude faster than brute force search.
Experiments confirm the sublinear dependence on $n$ and the linear dependence on the dimension.
We have compared against state-of-the-art LSH-based library {\tt FALCONN}: our search is somewhat slower, but memory usage and preprocessing time are significantly smaller.
\\\smallskip

\end{abstract}

\section{Introduction}

We are interested in the problem of Approximate Nearest Neighbor search 
in Euclidean spaces, when the dimension is high; typically one assumes for dimension $d\gg \log n$, where $n$ denotes the number of input data points. 
Due to known reductions, e.g.~\cite{HIM12}, it is apparent that one may focus on designing an efficient data structure for the Approximate Near Neighbor (ANN) problem instead of directly solving the Approximate Nearest Neighbor problem. 
The former is a decision problem, whose output may contain a witness point (as in Definition~\ref{Dann} below),
whereas the latter is an optimization question.
Here we trade exactness for efficiency in order to tackle the case of general dimension, where dimension is an input parameter.
The $(1+\epsilon,r)$-ANN problem, where $c=1+\epsilon$, is defined as follows.  
\begin{dfn}[Approximate Near Neighbor problem]\label{Dann}
Let $(\MM, \dd_{\MM})$ be a metric space. Given $P\subseteq \MM$, and reals $r>0$ and $\epsilon>0$, build a data structure s.t.\ for any query $q\in \MM$, there is an algorithm performing as follows:
 \begin{itemize}
  \item if $\exists p^{*} \in P$ s.t.\ $\dd_{\MM}(p^{*},q)\leq r$, then 
return any point $p'\in P$ s.t.\ $\dd_{\MM}(p',q)\leq(1+\epsilon)\cdot r$,
  \item if $\forall p \in P$, $\dd_{\MM}(p,q)>(1+\epsilon)\cdot r$, then report ``no''.
 \end{itemize}
\end{dfn}

An important approach for such problems today is Locality Sensitive Hashing (LSH).
It has been designed precisely for problems in general dimension.
The LSH method is based on the idea of using hash functions enhanced with the property that it is more probable to map nearby points to the same buckets.

\begin{dfn}
Let reals $r_1<r_2$ and $p_1>p_2>0$. We call a family $F$ of hash functions $(p_1,p_2,r_1,r_2)$-sensitive for a metric space $\MM$ if, for any $x, y \in \MM$, and $h$ distributed randomly in $F$, it holds:
\begin{itemize}
 \item $\mathrm{d}_{\MM} (x,y) \leq r_1 \implies Pr[h(x)=h(y)]\geq p_1,$
 \item $\mathrm{d}_{\MM} (x,y) \geq r_2 \implies Pr[h(x)=h(y)]\leq p_2.$
\end{itemize}
\end{dfn}

Let us survey previous work, focusing on methods whose complexity has polynomial, often even linear, dependence on the dimension.
LSH was introduced by Indyk and Motwani ~\cite{IM98,HIM12} and yields data structures with query time $O(dn^{\rho})$ and space $O(n^{1+\rho}+dn)$. Since then, the optimal value of the LSH exponent, $\rho$, has been extensively studied for several interesting metrics, such as $\ell_1$ and $\ell_2$. In a series of papers~\cite{IM98,DI04,MNP07,AI08,OWZ11}, it has been established that the optimal value for the Euclidean metric is $\rho={1}/{c^2}\pm o(1)$, whereas in~\cite{IM98, MNP07,OWZ11} that the optimal value for the Hamming distance is $\rho= {1}/{c}\pm o(1)$. 

In contrast with the definition above, which concerns data-independent LSH, quite recently the focus has been shifted to data-dependent LSH. 
In the latter case, the algorithms exploit the fact that every dataset has some structure and consequently this approach yields better bounds for the LSH exponent. Specifically, in~\cite{AINR14,AR15} they showed that $\rho={1}/({2c-1})$ for the $\ell_1$ distance and $\rho={1}/({2c^2-1})$ for the $\ell_2$ distance. Moreover, in~\cite{AR16} they proved that these bounds are tight in the data-dependent case. In a recent paper~\cite{ALRW17},
 the authors have established lower bounds and have also designed LSH-based algorithms to arrive at matching upper bounds in the data-dependent case. These results also yield data independent algorithms.
 
The data-dependent algorithms, though better in theory, are quite challenging in practice. In \cite{AILRS15}, they present an efficient implementation of one part of \cite{AR15}. 
Another attempt towards practicality for a data-dependent algorithm was recently made in~\cite{ARN17}, where they presented a new approach based on LSH forests. Typically though, data-independent algorithms, such as the one proposed in this work, yield better results in practice than data-dependent algorithms.

For practical applications, an equally important parameter is memory usage.
Most of the previous work in the (near) linear space regime focuses on the case that $c$ is greater than $1$ by a constant term.
When $c$ approaches $1$, these methods become trivial in the sense that query time becomes linear in $n$. 
One such LSH-based approach~\cite{Pan06} offers query time proportional to 
$dn^{O(1/{c})}$.
The query time was later improved in~\cite{AI08} to $dn^{O(1/{c^2})}$; both complexities are sublinear in $n$ only for large enough $c>1$.
One remarkable exception is the recently accepted paper~\cite{ALRW17}, where they achieve near-linear space and sublinear query time, even for $c\to 1^+$. Their data-dependent data structure is optimal for a reasonable model of hashing-based data structures. 

Another line of work that achieves linear space and sublinear query time for the Approximate Nearest Neighbor problem is based on random projections to drastically lower-dimensional spaces, where one can simply search using tree-based methods, such as BBD-trees~\cite{AEP15,AEP16}. This method relies on a projection extending the type of projections typically based on the Johnson-Lindenstrauss lemma. The new projection only ensures that an approximate nearest neighbor in the original space can be found among the preimages of $k$ approximate nearest neighbors in the projection. In the final version of this work~\cite{AEP16}, a sublinear-time algorithm is achieved with optimal space usage.

In this paper, a new random projection is again the key step, although here the projection's range are the vertices of the Hamming hypercube of dimension $\log n$. The projected space contains strings which serve as keys for buckets containing the input points. The query algorithm simply projects the query point, then examines points which are assigned to the same or nearby vertices on the Hamming cube. 

The random projection relies on the existence of an LSH family for the input metric. We study standard LSH families for $\ell_2$ and $\ell_1$, for which we achieve query time $O(d n^{1-\delta})$ where $\delta=\Theta(\epsilon^2)$, $\epsilon \in (0,1]$. The constants appearing in $\delta$ vary with the LSH family, but it holds that $\delta>0$ for any $\epsilon>0$. The space and preprocessing time are both linear for constant probability of success, which is important in practical applications.

We illustrate our approach with an open-source implementation, and report on a series of experiments with $n$ up to $10^6$ and $d$ up to 1000, where results are very encouraging. It is evident that our algorithm is 8.5--80 times faster than brute force search, depending on the difficulty of the dataset. Moreover, we handle a real dataset of $10^6$ images represented in 960 dimensions with a query time of less than $128$~msec on average. We test our implementation with synthetic and image datasets, and we compare against the LSH-based {\tt FALCONN} library \cite{AILRS15}. For SIFT, MNIST and GIST image datasets, we typically achieve 2.2 times less memory consumption and 14 times faster construction time, while {\tt FALCONN} is~1.15 times faster in query time.

The rest of the paper is structured as follows.
The next section states our main complexity results for the $(c,r)$-ANN problem in the Euclidean and Manhattan metrics. 
In Section~\ref{sec:impl}, we discuss our implementation, whereas in
Section~\ref{Sie}, we present our experimental results. 
We conclude with open questions.

\section{Data structures}\label{Sds}

This section introduces our main data structure, and the corresponding algorithmic tools. 

The algorithmic idea is to apply a random projection from any LSH-able metric to the Hamming hypercube.
Given an LSH family of functions for some metric space, we randomly project points to the Hamming cube of dimension $<=\log n$, where $n$ is the number of input points. Hence, we obtain binary strings serving as keys for buckets containing the input points. The query algorithm projects a given point, and tests points assigned to the same or nearby vertices on the hypercube. 

We start our presentation with an algorithmic idea applicable to any metric admitting an LSH-based construction, aka LSH-able metric. Then, we study some classical LSH families which are also simple to implement.

It is observed in~\cite{Cha02} that, for any similarity function that admits an LSH function family, the corresponding distance function is isometrically embeddable in the Hamming hypercube, i.e.\ the hypercube whose vertices have 0/1 coordinates. The dimension cannot be controlled, and depends on the similarity function. 
Our approach is to specify a random projection from any space endowed with an LSH-able metric, to the Hamming hypercube. Random projections which map points from a high dimensional Hamming space to lower dimensional Hamming space have been already used in the ANN context~\cite{KOR00}.
We start with an ANN data structure whose complexity and performance depends on the LSH family that we assume is available.

\begin{lem}\label{lem:DS}
Given a $(p_1,p_2,r,cr)$-sensitive hash family $F$ for some metric $(\MM,\dd_{\MM})$ and input dataset $P\subseteq \mathcal{M}$, there exists a data structure for the $(c,r)$-ANN problem with space $O(dn)$, time preprocessing $O(d n)$, and query time $O(dn^{1-\delta}+n^{H((1-p_1)/2)})$, where 
$$
\delta=\delta(p_1,p_2)=\frac{(p_1-p_2)^2}{(1-p_2)}\cdot \frac{\log \ee}{4},
$$
where $\ee$ denotes the basis of natural logarithms, and $H(\cdot)$ is the binary entropy function.
The bounds hold assuming that computing $\dd_{\MM}(.)$ and computing the hash function cost $O(d)$. Given some query $q\in \MM$, the building process succeeds with constant probability.
\end{lem}

\begin{proof}
The first step is a random projection to the Hamming space of dimension $d'$ for $d'$ to be specified later. 
We first sample $h_1\in F$. 
We denote by $h_1(P)$ the image of $h_1$ under $P$. 
Now for each element $x\in h_1(P)$, with probability $1/2$,  set $f_1(x)=0$, otherwise set $f_1(x)=1$. 
This is repeated $d'$ times, and eventually for $p\in \MM$ we compute $f(p)=(f_1(h_1(p)),\ldots,f_{d'}(h_{d'}(p)))$.
Now, observe that 
$$
\dd_{\mathcal{M}}(p,q)\leq r \implies \EE  [\|f_i(h_i(p))-f_i(h_i(q))\|_1]\leq 0.5 (1-p_1)\implies  \EE  [\|f(p)-f(q)\|_1]\leq 0.5\cdot d' \cdot (1-p_1),
$$
$$
\dd_{\mathcal{M}}(p,q)\geq cr \implies \EE  [\|f_i(h_i(p)-f_i(h_i(q))\|_1]\geq 0.5 (1-p_2) \implies  \EE  [\|f(p)-f(q)\|_1]\geq 0.5\cdot d'\cdot (1-p_2).
$$
We distinguish two cases.

First, consider the case $\dd_{\mathcal{M}}(p,q)\leq r$. Let $\mu=\EE[\|f(p)-f(q)\|_1]$. {Then},
$$
\Pr[\|f(p)-f(q)\|_1\geq \ \mu]\leq \frac{1}{2},
$$
since $\|f(p)-f(q)\|_1$ follows the binomial distribution. 

Second, consider the case 
$\dd_{\mathcal{M}}(p,q)\geq cr$. By typical Chernoff bounds,
$$
Pr[\|f(p)-f(q)\|_1\leq \frac{1-p_1}{1-p_2} \cdot \mu]\leq exp(-0.5 \cdot \mu \cdot (p_1-p_2)^2/(1-p_2)^2)\leq exp(- d' \cdot (p_1-p_2)^2/4(1-p_2)).
$$

After mapping the query $q\in \mathcal{M}$ to $f(q)$ in the $d'$-dimensional hamming space we search for all ``near'' hamming vectors $f(p)$ s.t. $\|f(p)-f(q)\|_1\leq 0.5 \cdot d' \cdot (1-p_1)$. This search costs 
$$
\binom{d'}{1}+\binom{d'}{2}+\cdots+\binom{d'}{\lfloor d' \cdot (1-p_1)/2 \rfloor}\leq O(d'\cdot 2^{d' \cdot H((1-p_1)/2)}),
$$ 
where $H(\cdot)$ is the binary entropy function. The inequality is obtained from standard bounds on binomial coefficients (see e.g. \cite{MU05}).
Now, the expected number of points $p\in P$ for which $\dd_{\mathcal{M}}(p,q)\geq cr$ but are mapped "near" $q$ is $\leq n\cdot exp(- d' \cdot (p_1-p_2)^2/4(1-p_2)))$. If we set 
$
d'= \log n,
$
we obtain expected query time 
$$
O(n^{H((1-p_1)/2))}+d  n^{1-\delta}),
$$
where 
$$
\delta=\frac{(p_1-p_2)^2}{(1-p_2)}\cdot \frac{\log \ee}{4}.
$$
If we stop searching after having seen, say $10 n^{1-\delta}$ points for which $\dd_{\mathcal{M}}(p,q)\geq cr$ , then we obtain the same time with constant probability of success. Notice that "success" translates to successful preprocessing for a fixed query $q \in \MM$. The space required is $O(dn)$. 
\end{proof}

The lower bound on $\delta$ is not tight. Above, we have used simplified Chernoff bounds to keep our exposition simple. 
\paragraph*{Discussion on parameters.} We set the dimension $d'=\log n$, since it  minimizes the expected number of candidates
under the linear space restriction. We note that it is possible to set $d'<\log n$ and still have sublinear query time. This choice of $d'$ is interesting in practical applications since it improves space requirement.  
\subsection{The $\ell_2$ case}

In this subsection, we consider the $(c,r)$-ANN problem when the dataset consists of $n$ points $P\subset \RR^d$, the query is $q\in \RR^d$, and the distance is the Euclidean metric. 

We may assume, without loss of generality, that $r=1$, since we can uniformly {scale $(\RR^d,\|\cdot\|_2)$}.
We shall consider two LSH families, for which we obtain slightly different results. The first is based on projecting points to random lines, 
and it is the algorithm used in our implementation, which is discussed in Section~\ref{sec:impl}. The second LSH family relies on reducing the Euclidean problem to points on the sphere, and then partitioning the sphere with random hyperplanes.

\subsubsection{Project on random lines}\label{SSSlines}

Let $p$, $q$ two points in ${\RR}^d$ and $\eta$ the distance between them. Let $w>0$ be a real parameter, and
let $t$ be a random number distributed uniformly in the interval $[0, w]$.
In \cite{DI04}, they present the following LSH family. For $p \in {\RR}^d$, consider the random function 
$$
h(p)=\left \lfloor \frac{\langle p, v\rangle+t}{w} \right \rfloor ,\quad p,v\in\RR^d,
$$
where $v$ is a vector randomly distributed with the $d$-dimensional normal distribution. For this LSH family, the probability of collusion is 
$$
\alpha({\eta,w})= \int_{t=0}^w \frac{2}{\sqrt{2\pi}\eta}\exp{\Big(-\frac{t^2}{2{\eta}^2}\Big)}\Big(1-\frac{t}{w}\Big)dt. $$

\begin{lem}\label{lem:lines}
Given a set of $n$ points $P\subseteq \RR^d$, there exists a data structure for the $(c,r)$-ANN problem under the Euclidean metric, requiring space $O(dn)$, time preprocessing {$O(dn)$}, and query time $O(dn^{1-\delta}+n^{0.9})$, where 
$$
\delta \geq 0.03 \, {(c-1)}^2.
$$
Given some query point $q\in \RR^d$, the building process succeeds with constant probability.
\end{lem}

\begin{proof}
In the sequel we use the standard Gauss error function, denoted by $erf(\cdot)$.
For probabilities $p_1,p_2$, it holds that
$$
p_1=\alpha({1,w})= \int_{t=0}^w \frac{2}{\sqrt{2\pi}}\exp{\Big(-\frac{t^2}{2}\Big)}\Big(1-\frac{t}{w}\Big)dt =
erf\Big(\frac{w}{\sqrt{2}}\Big)-\sqrt{\frac{2}{\pi}}\frac{1}{w}\Big(1-\exp\big(-\frac{w^2}{2}\big)\Big),
$$
and also that
$$
p_2=\alpha({c,w})= \int_{t=0}^w \frac{2}{\sqrt{2\pi}c}\exp{\Big(-\frac{t^2}{2{c}^2}\Big)}\Big(1-\frac{t}{w}\Big)dt =
$$ 
$$
= erf\Big(\frac{w}{\sqrt{2}c}\Big)-\sqrt{\frac{2}{\pi}}\frac{c}{w}\Big(1-\exp\big(-\frac{w^2}{2{c}^2}\big)\Big).
$$
The LSH scheme is parameterized by $w$.
One possible value is $w=3$, as we have checked on a computer algebra system.
On the other hand, $w=c$ gives similar results, and they are simpler to obtain.
In particular, we have 
$$
p_1-p_2 = erf\Big(\frac{c}{\sqrt{2}}\Big) -\sqrt{\frac{2}{\pi}}\frac{1}{c}\Big(1-\exp\big(-\frac{c^2}{2}\big)\Big)- erf\Big(\frac{1}{\sqrt{2}}\Big)+\sqrt{\frac{2}{\pi}}\Big(1-\exp\big(-\frac{1}{2}\big)\Big).
$$
We shall prove that, given $w=c$, for $c \in (1,2]$, it holds that $p_1-p_2>\frac{5(c-1)}{21}$.
Let us define 
$$
g(c)=p_1-p_2-\frac{5(c-1)}{21} = erf\Big(\frac{c}{\sqrt{2}}\Big) -\sqrt{\frac{2}{\pi}}\frac{1}{c}\Big(1-\exp\big(-\frac{c^2}{2}\big)\Big)-
$$
$$
-erf\Big(\frac{1}{\sqrt{2}}\Big)+\sqrt{\frac{2}{\pi}}\Big(1-\exp\big(-\frac{1}{2}\big)\Big)- \frac{5(c-1)}{21},
\quad c\in (1,2].
$$
Using elementary calculus, it is easy to show that $g(c)$ is concave over $c\in (1,2]$.
Also, $g(1)=0$ and $g(2)>0$, thus $\forall c \in(1,2], \enspace g(c)>0$ and consequently $p_1-p_2>\frac{5(c-1)}{21}$.
In addition, $w=c$ implies $1-p_2 = 1-erf\big(\frac{1}{\sqrt{2}}\big)+\sqrt{\frac{2}{\pi}}\big(1-\exp(-\frac{1}{2})\big) < 0.64$,
and $H\big(\frac{1-p_1}{2}\big)<{0.9}$.
Hence, for $w=c$ and $c \in(1,2]$, $\delta > 0.03{(c-1)}^2$.
\end{proof}


\subsubsection{Hyperplane LSH}

This section reduces the Euclidean ANN to an instance of ANN for which the points lie on a unit sphere. The latter admits an LSH scheme based on partitioning the space by randomly selected halfspaces.

In Euclidean space $\RR^d$, let us assume that the dimension is $d=O(\log n \cdot \log \log n)$, since one can project points \`{a} la Johnson Lindenstrauss \cite{DG02}, and preserve pairwise distances up to multiplicative factors of $1\pm o(1)$. Then, we partition $\RR^d$ using a randomly shifted grid, with cell edge of length 
$O{(\sqrt{d})}=O{((\log n \cdot \log \log n)^{1/2})}$.
Any two points $p,q\in \RR^d$ for which 
$\|p-q\|_2\leq 1$ lie in the same cell with constant probability. 
Let us focus on the set of points lying inside one cell. This set of points has diameter bounded by $O{((\log n \cdot \log \log n)^{1/2})}$. Now, a reduction of \cite{Val15}, reduces the problem to an instance of ANN for which all points lie on a unit sphere $\Ss^{d-1}$, and the search radius is roughly $r'=\Theta((\log n\cdot \log \log n)^{-1/2})$. These steps have been also used in \cite{ALRW17}, as a data-independent reduction to the spherical instance. 

We now consider the LSH family introduced in~\cite{Cha02}. Given $n$ unit vectors $P\subset \Ss^{d-1}$, we define, for each $q \in \Ss^{d-1}$, hash function
$h(q)=\sign \langle q,v\rangle$, where $v$ is a random unit vector. Obviously, 
$\Pr[h(p)=h(q)]=1-\frac{\theta(p,q)}{\pi}$, 
where $\theta(p,q)$ denotes the angle formed by the vectors $p\ne q\in\Ss^{d-1}$.
Instead of directly using the family of~\cite{Cha02}, we employ its amplified version, obtained by concatenating $k\approx 1/r'$ functions $h(\cdot)$, each chosen independently and uniformly at random from the underlying family.
The amplified function $g(\cdot)$ shall be fully defined in the proof below.
This procedure leads to the following.

\begin{lem}\label{lem:hplanes}
Given a set of $n$ points $P\subset \RR^d$, there exists a data structure for the $(c,r)$-ANN problem under the Euclidean metric, requiring space {$O(dn)$}, time preprocessing {$O(dn)$}, and query time $O(dn^{1-\delta}+n^{0.91})$, where 
$$
\delta \geq 0.05 \cdot \Big(\frac{c-1}{{c}} \Big)^2.
$$ 
Given some query $q\in \RR^d$, the building process succeeds with constant probability.
\end{lem}

\begin{proof}
We exploit the reduction described above that translates the Euclidean ANN to a spherical instance of ANN with search radius $r'=\Theta((\log n\cdot \log \log n)^{-1/2})$.
The latter is handled by a hyperplane LSH scheme based on~\cite{Cha02} as detailed immediately below.

Let us denote by $F$ the aforementioned LSH family of~\cite{Cha02}. We build a new (amplified) family of functions $G_k=\{g(x)=(h_1(x),\ldots,h_k(x)) : i=1\ldots k,\, h_i\in F)\}$. Now, obviously, for any two unit vectors $p\ne q$, we have
$$
\Pr_{g\in G}[g(p)=g(q)]=\Big(1-\frac{\theta(p,q)}{\pi}\Big)^k.
$$
Hence,
$$
\|p-q\|_2\leq r' \implies 2 \sin \Big(\frac{\theta(p,q)}{2}\Big)\leq r'\implies \theta(p,q) \leq 2 \arcsin \Big(\frac{r'}{2} \Big)=\theta_r,
$$
which defines $\theta_r$. Moreover,
$$
\|p-q\|_2\geq c r'\implies  2 \sin \Big(\frac{\theta(p,q)}{2}\Big) \geq c r' \implies
\theta(p,q)\geq 2\arcsin\Big( \frac{cr'}{2}\Big).
$$
By using elementary calculus, 
it is easy to prove that
$$
2\arcsin\Big( \frac{cr'}{2}\Big) \geq 2 c\cdot\arcsin\Big( \frac{r'}{2}\Big)\implies \theta(p,q)\geq c\cdot \theta_r.
$$
Hence, for 
$k=\lfloor\pi/\theta_r \rfloor$ and since $r'=\Theta((\log n \cdot \log \log n)^{-1/2})\implies \theta_r=o(1)$, 
$$
p_1=\Pr[g(p)=g(q) \mid \|p-q\|_2\leq r]\geq \Big(1-\frac{\theta_r}{\pi} \Big)^k\geq exp(-\frac{\pi}{(\pi-\theta_r)})\geq \frac{1}{\ee^{1+o(1)}},
$$
$$
p_2=\Pr[g(p)=g(q) \mid \|p-q\|_2\geq c \cdot r]\leq \Big(1-\frac{c \cdot \theta_r}{\pi} \Big)^k\leq 
exp(-\frac{c \theta_r}{\pi}\cdot(\frac{\pi}{\theta_r}-1))\leq \frac{1}{{c} \cdot \ee^{1-o(1)}}.
$$
Now applying Lemma~\ref{lem:DS} yields 
$$
\delta \geq \frac{1}{\ee^{2+o(1)}}\cdot\Big( 1-\frac{\ee^{o(1)}}{{c}} \Big)^2\cdot \frac{1}{1-(c\cdot\ee)^{-1}} \cdot \frac{\log(\ee)}{4}
\geq 0.059\cdot  \Big(1-\frac{1}{{c}}\Big)^2,
$$
for $c\in (1,2]$.
The space required is 
$O(dn+n(d'+k))$ where $d'$ is defined in Lemma \ref{lem:DS}. 
Since $k\ll d$, and $d' \ll d$ the total space is $O(dn)$. Notice also that $H(\frac{1-p_1}{2})\leq 0.91$.
\end{proof}
The data structure of Lemma~\ref{lem:hplanes} provides slightly better query time than that of Lemma~\ref{lem:lines}, when $c$ is small enough. 
\subsection{The $\ell_1$ case}

In this section, we study the $(c,r)$-ANN problem under the ${\ell}_1$ metric. The dataset consists again of $n$ points $P\subset {\RR}^d$ and the query point is $q \in {\RR}^d$.

For this case, let us consider the following LSH family, introduced in~\cite{AI06}. In particular, a point $p$ is hashed as follows:
$$
h(p)=\Big(\left \lfloor \frac{p_1+t_{1}}{w} \right \rfloor , \left \lfloor \frac{p_2+t_{2}}{w} \right \rfloor, \ldots, \left \lfloor \frac{p_{d}+t_{d}}{w} \right \rfloor \Big) ,
$$
where $p=(p_1, p_2, \ldots, p_{d})$ is a point in $P$, $w=\alpha r$, and the $t_{i}$ are drawn uniformly at random from $[ 0,\ldots ,w)$. Buckets correspond to cells of a randomly shifted grid.

Now, in order to obtain a better lower bound, we employ an amplified hash function, defined by concatenation of $k=\alpha$ functions $h(\cdot)$ chosen uniformly at random from the above family.

\begin{lem}
Given a set of $n$ points $P\subseteq \RR^d$, there exists a data structure for the $(c,r)$-ANN problem under the $\ell_1$ metric, requiring space $O(dn)$, time preprocessing $O(dn)$, and query time $O(dn^{1-\delta}+n^{0.91})$, where 
$$
\delta \geq 0.05 \cdot \Big(\frac{c-1}{{c}} \Big)^2.
$$
Given some query point $q\in \RR^d$, the building process succeeds with constant probability.
\end{lem}

\begin{proof}
We denote by $F$ the previously introduced LSH family of~\cite{AI06}, which is 
$(1-\frac{1}{\alpha},1-\frac{c}{c+\alpha},1,c)$-sensitive. We build the amplified family of functions $G_k=\{g(x)=(h_1(x),\ldots,h_k(x)) : i=1,\ldots, k,~ h_i\in F) \}$. 
Setting $\alpha = k = \log{n}$, we have:
$$
{p_1}={\Big(1-\frac{1}{\alpha}\Big)}^k={\Big(1-\frac{1}{\log{n}}\Big)}^{\log{n}} \geq {\Big(\exp \Big(-\frac{1}{\log n-1}\Big)\Big)}^{\log{n}}\geq \frac{1}{\ee^{1+o(1)}},
$$ 
$$
{p_2}={\Big(1- \frac{c}{\alpha+{c}}\Big)}^k = {\Big(1- \frac{c}{\log n+c}\Big)}^{\log n} .
$$
Hence, 
$$
{p_2} \geq \exp(-c) \geq \frac{1}{\ee \cdot (2c-1)},
$$
and 
$$
{p_2} \leq \exp \Big(-\frac{c}{1+\frac{c}{\log n}}\Big)=  \exp \Big(-\frac{c}{1+o(1)}\Big) \leq \exp \big(-c+o(1)\big) \leq \frac{\ee^{o(1)}}{\ee c}.
$$
Therefore, for $n$ large enough, it holds that 
$$ \delta = \frac{{\big({p_1}-{p_2} \big)}^2}{\big(1-{p_2} \big)}\cdot \frac{\log \ee}{4} \geq \frac{1}{\ee^{2+o(1)}} \cdot \frac{{(1-\frac{1}{c}) }^2}{1-\frac{1}{\ee (2c-1)}} \cdot \frac{\log \ee}{4}  
\geq 0.055 \cdot (1-\frac{1}{c})^2,
$$
for $c\in(1,2]$. Notice that $H((1-p_1)/2)\leq 0.91$.
\end{proof}

\section{Implementation}\label{sec:impl}

This section discusses our C++ library, named {\tt Dolphinn}, which is available online\footnote{\url{https://github.com/gsamaras/Dolphinn}}.

The project is open source, under the BSD 2-clause license. The code has been compiled with {\tt g++ 4.9} compiler, with the O3 optimization flag enabled. It contains only 716 lines of code.
Important implementation issues are discussed here, focusing on efficiency. 

Our implementation is based on the algorithm from Subsection~\ref{SSSlines} and supports similarity search under the Euclidean metric. We denote by $F$ the LSH family introduced above (see also~\cite{DI04}) and discussed in Subsection~\ref{SSSlines}. 
The data structure to which the points are mapped is called Hypercube.


\subsubsection*{Parameters.} Our implementation provides several parameters to allow the user to fully customize the data structure and search algorithm, such as:

\begin{itemize}

	\item[$d'$] The dimension of the Hypercube to which the points are going to be mapped. The larger this parameter, the faster the query time, since it leads to a finer mapping of the points. However, this speedup comes with increased memory consumption and build time.

	\item[$Threshold$] Maximum number of points to be checked during search phase. The greater this parameter, the more accurate results will be produced, at the cost of an increasing query time.

	\item[$r$] It is provided by the user as input. The algorithm checks whether a point lies within a radius $r$ from the query point. 
	
	\item[$w$] Indexes a hash function $h$ from LSH family $F$. It specifies the distribution used by our data structure to map points to the 0/1 vertices of the Hypercube.

	\item[$\mu$, $\sigma$] Index a hash function $h$ from LSH family $F$, and specify the random vector under the Normal Distribution $\mathcal{N}(\mu,\sigma^2)$, which is multiplied by the point vector. It determines the distribution used by our data structure to map points to the 0/1 vertices of the Hypercube.
\end{itemize}


\subsubsection*{Configuration.} Despite the simple parameter set, we employ default values for all parameters (except $r$), which are shown to make {\tt Dolphinn} run efficiently and accurately for several datasets. Moreover, the whole configuration is desgined for an one-threaded application. However, it would be interesting to modify it and take advantage of the huge parallel potential of {\tt Dolphinn}, both in the preprocessing and search stage.


\subsubsection*{Hypercube.}
The Hypercube data structure contains its statistical information, so that it can efficiently hash a query on arrival, by using two hash tables, that indicate the statistical choices made upon build and the index of every associated original point. Moreover, another hashtable is constructed, which maps the vertices of the Hypercube to their assigned original points. This hierarchy provides space and time efficiency, and is very natural to code, because of the simplicity of the algorithm.

\section{Experimental results}\label{Sie}

\setlength{\intextsep}{16pt}
\newcommand{\squeeze}{\vspace{-20pt}}

This section presents our experimental results and comparisons on a number of synthetic and real datasets. 

All experiments are conducted on a processor at 3~GHz$\times 4$ with 8~GB. We compare with the state-of-the-art LSH-based {\tt FALCONN}\footnote{\url{https://falconn-lib.org/}} library, and to the brute force approach, focusing on build and search times, as well as memory consumption and accuracy.

\subsubsection*{Datasets.} We use five datasets of varying dimensionality and cardinality. To test special topologies, the first two sets, namely \emph{Klein bottle} and \emph{Sphere}, are synthetic. We generate points on a Klein bottle and on a sphere embedded in $\RR^d$, then add to each coordinate zero-mean Gaussian noise of standard deviation~0.05 and~0.1 respectively. In both cases, queries are nearly equidistant to all points, which implies high miss rates.

The other three datasets, \emph{MNIST}\footnote{\url{http://yann.lecun.com/exdb/mnist/}}, \emph{SIFT} and \emph{GIST}\footnote{\url{http://corpus-texmex.irisa.fr/}} are presented in~\cite{JeDS11}, and are very common in computer vision, image processing, and machine learning. SIFT is a 128-dimensional vector that describes a local image patch by histograms of local gradient orientations. MNIST contains vectors of 784 dimensions, that are $28\times 28$ image patches of handwritten digits. There is a set of $60$k vectors, plus an additional set of $10$k vectors that we use as queries. GIST is a 960-dimensional vector that describes globally an entire image. SIFT and GIST datasets each contain one million vectors and an additional set for queries, of cardinality $10^4$ for SIFT and 1000 for GIST. Small SIFT is also examined, with $10^4$ vectors and 100 queries.

For the synthetic datasets, we solve the one near neighbor problem with a fixed radius of $1$ and we compare to brute force, since 
{\tt FALCONN} does not provide such a method yet.
Queries are constructed like this: Uniformly randomly pick a point from the point set and, following a Normal Distribution $N(0, 1)$, add a small quantity to every coordinate. A query is chosen to lie within the fixed radius $1$ with a probability of $50\%$. For the image datasets, we find all near neighbors within a fixed radius of $1$, by modifying our algorithm in order not to stop when it founds a point, if any, that lies within the given radius, but to continue until the threshold of points to be checked is reached. Moreover, we tune the number of probes for multiprobe LSH used by {\tt FALCONN}. For fair comparison, both implementations are configured in a way that yields the same accuracy.


\subsubsection*{Preprocessing.} 
The preproccesing phase of {\tt Dolphinn} is not a heavy operation and applications that tend to build several data structures for different point sets can benefit from this. Moreover, the preprocessing time of {\tt Dolphinn} has a linear dependence in $n$ and $d$, as expected, which is shown in Tables~\ref{tab:fixndSphere} and~\ref{tab:fixndKlein}.

{\tt Dolphinn} is expected 
to be faster than {\tt FALCONN} when building the data-structure, and Table~\ref{tab:build} shows representative experiments where {\tt Dolphinn} yields an average of~6.5 times faster preprocessing time. One fact that explains this significant speedup is the major difference on memory consumption by the different implementations, as discussed below.
Note that any normalization and/or centering of the point set, which is a requirement for {\tt FALCONN}, is {\it not} taken into account when counting runtimes. 

Table~\ref{tab:build} also shows the dependence of build time on the number of hashtables for {\tt FALCONN} and dimension $d'$ for {\tt Dolphinn}. The former is a linear relationship, whereas our method has roughly exponential complexity in $d'$.
Nonetheless, {\tt Dolphinn} is significantly faster.

\subsubsection*{Search and Memory.} 
We conduct experiments on our synthetic datasets while keeping $n$ or $d$ fixed, in order to illustrate how our algorithm's complexity scales in practice. The Sphere dataset is easier than the Klein bottle, which explains the reduced accuracy, as well as the dramatic decrease of speedup (w.r.t.\ brute force), since more points are likely to lie within a fixed radius of $1$, than in the Klein bottle. In general, our algorithm seems to scale well, sublinearly in $n$ and linearly in $d$, as shown in Table~\ref{tab:fixndSphere} and Table~\ref{tab:fixndKlein}.

For {ANN search}, we introduce a threshold on the number of points that the algorithm actually checks, in other words, on the maximum allowed length of the final candidate list. While preserving a good accuracy rate (i.e. above 90\%), we observe how the search time drops proportionally to the threshold, as shown in Table~\ref{tab:maxPntsSphere} and Table~\ref{tab:maxPntsKlein}. We observe that the search time grows linearly (in worst case scenario\footnote{negative answer for the near neighbor decision problem}), as expected, while the accuracy increases as the threshold increases, as shown in Table~\ref{tab:maxPntsSIFT}. However, one should be aware that the linear dependence depends on whether the algorithm visits empty buckets or not.
In practice, the relation between the accuracy and the threshold should be determined by the programmer, by experimenting and finding a sweet spot.

Moreover, We report query times between {\tt FALCONN} and {\tt Dolphinn} on the image datasets; small SIFT, SIFT, MNIST and GIST, as well as memory consumption and building times. {\tt Dolphinn} is 21.5 times faster than {\tt FALCONN} in the preprocessing stage. Space consumption of {\tt Dolphinn} is much better in the Small SIFT experiment, while {\tt FALCONN} is a bit faster in search. In all other cases, {\tt Dolphinn} consumes 2.1 times less memory, while {\tt FALCONN} is 1.15x faster in search, as shown in Table~\ref{tab:search}.

\begin{table} \begin{center}
\begin{tabular}{ |c|c||r|r|r| } \hline
	\multirow{1}{1em}{$n$} & \multirow{1}{1em}{$d$} & \multicolumn{1}{c|}{build} & \multicolumn{1}{c|}{search ($\mu$sec)}  & \multicolumn{1}{c|}{brute f.~search (sec)} \\
	\hline\hline
	\multirow{5}{2em}{$10^5$} & 128 	& 0.053	& 62.37    &	0.006	\\
	                          		 &  256	& 0.092	& 152.3	&	0.012	\\
	                          		 &  512	& 0.168	& 257.1	&	0.025	\\
	                           		 &  800	& 0.255	& 374.1	&	0.039	\\
	                          		 & 1024	& 0.321	& 499.6	&	0.050	\\ \hline
	100		&	\multirow{5}{2em}{$512$}		& 0.0002	& 1.001	& 7.5e-05	\\
	1000		&							& 0.0016	& 4.924	& 0.0004	\\
	10000	&                           				& 0.0169	& 47.72	& 0.0049	\\
	100000	&                           				& 0.1683	& 477.0	& 0.0499	\\
	1000000	&                           				& 1.6800	& 2529	& 0.2492	 \\ \hline
\end{tabular}
\caption{Sphere build and search for varying $n,d$, where the other parameter is fixed. This experiment confirms that the build time, as well as the search time, scale gracefully along with $n,d$, while accuracy remained maximum and the speedup over brute force was 80 times on average.\label{tab:fixndSphere}}
\end{center}
\squeeze
\end{table}

\begin{table} \begin{center}
\begin{tabular}{ |c|c||r|r|r| } \hline
	\multirow{1}{1em}{$n$} & \multirow{1}{1em}{$d$} & \multicolumn{1}{c|}{build} & \multicolumn{1}{c|}{search (sec)}  & \multicolumn{1}{c|}{brute f. search (sec)} \\
	\hline\hline
	\multirow{5}{2em}{$10^5$} & 128 	& 0.053	& 0.0009   &	0.0061	\\
	                          		 &  256	& 0.091	& 0.0029	&	0.0147	\\
	                          		 &  512	& 0.168	& 0.0031	&	0.0254	\\
	                           		 &  800	& 0.259	& 0.0056	&	0.0425	\\
	                          		 & 1024	& 0.321	& 0.0061	&	0.0513	\\ \hline
	100		&	\multirow{5}{2em}{$512$}		& 0.0003	& 1e-05	& 7e-05	\\
	1000		&							& 0.0016	& 4e-05	& 0.0004	\\
	10000	&                           				& 0.0169	& 0.0004	& 0.0049	\\
	100000	&                           				& 0.1679	& 0.0051	& 0.0501	\\
	1000000	&                           				& 1.6816	& 0.0252	& 0.2497	 \\ \hline
\end{tabular}
\caption{Klein bottle build and search for varying $n,d$, where the other parameter is fixed. This experiment confirms that the build and search times scale nicely with $n,d$, while the average accuracy was 98.8\%. The speedup over brute force was 8.5x on average.\label{tab:fixndKlein}}
\end{center}
\squeeze
\end{table}


\begin{table}\begin{center}
\begin{tabular}{|c|rrrrrrrr| } \hline
Threshold   &1 &	100 &	300 &	500 &	700 &	1000 &	5000 &	10000 \\ \hline\hline
		search ($\mu$s)     & 0.64 &	24.5 &	73.1 &	110 &	170 &	260 &	1348 &	3207 \\ \hline
\end{tabular}
\caption{Sphere ($n = 10^6, d = 512$) was the point set; it is shown how the number of points that must be actually checked increases as a function of {Threshold}. \label{tab:maxPntsSphere}}
\end{center}
\squeeze
\end{table}


\begin{table}\begin{center}
\begin{tabular}{|c|rrrrr|} \hline
	Threshold	     &10 		& 100 	& 1000  &	5000 & 10000  \\ \hline\hline
		search ($\mu$s)    & 6.688 	& 47.78	& 615.2 &	2542 & 7754 \\ \hline
\end{tabular}
\caption{Klein bottle ($n = 10^5, d = 512$) was the point set; it is shown how the number of points that must be actually checked increases as a function of {Threshold}. \label{tab:maxPntsKlein}}
\end{center}
\squeeze
\end{table}
\begin{table}\begin{center}
\begin{tabular}{|c|rrrrrrr| } \hline
	Threshold	&1		&10 		& 100 	& 1000   	&	10000	& 100000	& 1000000 \\ \hline\hline
		search		& 0.0002	& 0.0013 	& 0.0139	& 0.1884	&	2.707	& 16.73  	& 17.75 \\
		accuracy\%       &0		& 0	 	& 0		& 0.26 	&	7.19 		& 80.49  	& 100 \\ \hline
\end{tabular}
\caption{SIFT was used as the point set, while displaying how the number of points to actually check affects search time (msec) and accuracy. \label{tab:maxPntsSIFT}}
\end{center}
\squeeze
\end{table}


\begin{table}\begin{center}
\begin{tabular}{|c|cccc|cccc|cccc|} \hline
  & \multicolumn{4}{c|}{SIFT $n = 10^4, d = 128$} & \multicolumn{4}{c|}{SIFT $n = 10^6, d = 128$} & \multicolumn{4}{c|}{GIST $n = 10^6, d = 960$} \\ \hline
$\#$hashtables or $d'$ &  2  	    &  4      &  8      &  16    & 2      &  4     &  8     &  16   & 2      &  4     &  8      &  16  \\ \hline\hline
{\tt FALCONN} 	&  0.11 &  0.12 &  0.17 &  0.32 & 3.58 & 3.64 & 7.27 & 14.5 & 18.2 & 18.9 & 39.3  & 75.7  \\
{\tt Dolphinn}        	&  0.01 &  0.02 &  0.02 &  0.05 & 0.52 & 0.52 & 1.49 & 3.33 & 2.03 &  4.01 &  7.98 &  15.9 \\ \hline
\end{tabular}
\caption{Build time (s) for 3 representative datasets, where it is evident that {\tt Dolphinn} has~6.5 times faster construction time than {\tt FALCONN}; $d'$ is the dimension of Hypercube, while $\#$hashtables is a critical parameter used by {\tt FALCONN} for multi-probe LSH.\label{tab:build}}
\end{center}
\squeeze
\end{table}

\hspace*{-1cm}\begin{table}\begin{center}
\begin{tabular}{|l|l|r|l|r|l|r|r|r|} \hline
       & \multicolumn{2}{c|}{Small SIFT} & \multicolumn{2}{c|}{SIFT} & \multicolumn{2}{c|}{MNIST} &\multicolumn{2}{c|}{GIST} \\ \hline
   &  {\tt FALCONN}     &  {\tt Dolphinn} 	& {\tt FALCONN} 	&  {\tt Dolphinn}  &  {\tt FALCONN}     	&  {\tt Dolphinn}  	&  {\tt FALCONN}     	&  {\tt Dolphinn}       \\ \hline\hline
Build   	&  0.241    		&  0.010   &  12.71    		&  0.571   	& 3.426     		&  0.158	&  55.69     		&  3.023      \\
Memory  	&  0.064     		&  0.008    &  1.167    		&  0.526   &  0.429     		&  0.219	&  7.666   			&  3.776      \\
Search	&  7e-05     		&  9e-05   &  0.007     		&  0.009   &  0.0005     		&  0.0005		&  0.123     		&  0.128         \\ \hline
\end{tabular}
\caption{Memory usage (gb), and Build and Search times (s) for the image datasets. {\tt Dolphinn} is times faster than {\tt FALCONN} in the preprocessing stage. Space consumption of {\tt Dolphinn} is much better in the Small SIFT experiment, while {\tt FALCONN} is a bit faster in search. In all other cases, {\tt Dolphinn} consumes 2.1 times less memory, while {\tt FALCONN} is 1.15x faster in search.
\label{tab:search}}
\end{center}
\end{table}


\section{Conclusion}

We have designed a conceptually simple method for a fast and compact
approach for Near Neighbor queries, and have tested it experimentally. This offers a competitive approach to Approximate Nearest Neighbor search.

Our method is optimal in space, with sublinear query time for any constant approximation factor $c>1$. 
The algorithm randomly projects points to the Hamming hypercube. The query algorithm simply projects the query point, then examines points which are assigned to the same or nearby vertices on the hypercube. 
We have analyze the query time for the Euclidean and Manhattan metrics. 

We have focused only on data-independent methods for ANN. However, data-dependent methods are known to perform better in theory. Hence, designing a practical data-dependent variant of our method will be a very interesting next step. 

\bibliographystyle{alpha}
\bibliography{jlann}

\end{document}